\documentclass[10pt]{bmc_article}
\usepackage{amsmath,epsfig,amssymb,enumerate,tabularx,empheq,amsthm,appendix}
\newtheorem{theorem}{Theorem}
\newtheorem{lemma}[theorem]{Lemma}

\newtheorem{corollary}[theorem]{Corollary}
\newtheorem{assumption}{{\bf Assumption}}

\newcommand{\transconj}[1]{\ensuremath{{#1}^{\mathrm{\scriptscriptstyle{H}}}}}
\newcommand{\transpose}[1]{\ensuremath{{#1}^{\mathrm{\scriptscriptstyle{T}}}}}

\newtheorem{remark}{Remark}
\usepackage{tikz}
\usepackage{pgfplots}

\renewcommand{\P}{{\bf P}}

\newcommand{\V}{{\bf V}}
\newcommand{\D}{{\bf D}}

\newcommand{\Y}{{\bf Y}}
\newcommand{\W}{{\bf W}}

\newcommand{\J}{{\bf J}}

\renewcommand{\H}{{\bf H}}
\newcommand{\I}{{\bf I}}

\DeclareMathOperator{\tr}{tr}

\title{Performance Analysis and Optimal Power Allocation for Linear Receivers Based on Superimposed Training }
\author{{Abla Kammoun\correspondingauthor$^{1}$ \email{abla.kammoun@supelec.fr}, Karim Abed-Meraim \email{karim.abed@telecom-paristech.fr}}\\}
\address{%
	\iid(1)  Alcatel Lucent chair, Sup\'elec France \\
	\iid(2) TSI, T\'el\'ecom-Paris Tech
}%
\begin{document}
\maketitle
\begin{abstract}
In this paper, we derive a performance comparison between two training-based schemes for Multiple-Input Multiple-Output (MIMO) systems. The two schemes are the time-division multiplexing scheme and the  recently proposed  data-dependent superimposed pilot scheme. For  both schemes, a closed-form expressions for the Bit Error Rate (BER) is provided.  We also determine, for both schemes, the  optimal allocation of power between pilot and data that minimizes the BER. 

{\bf Key words:} superimposed trai\-ning sequence, MIMO systems performance, linear receiver.
\end{abstract}
\section{Introduction}

The use of Multiple-Input Multiple-Output (MIMO) antenna systems enables  high data rates without any increase in bandwidth or power consumption. However, the good performance of MIMO systems requires  a priori knowledge of the channel at the receiver. In many practical systems, the receiver estimates the channel by time division multiplexing pilot symbols with the data. Although high quality of channel estimation could be achieved especially when using a large number of pilot symbols \cite{hassibi}, this method may entail a waste of the available channel ressources. An alternative method is the conventional superimposed training. It consists in transmitting pilots and data at the same time. However, since during channel estimation, data symbols act as a  source of noise, channel estimation is  affected. In the literature, the impact of channel estimation error upon the performance indexes has been investigated. In \cite{bohlin} and \cite{bohlin1}, a comparison between the performance of the conventional superimposed training scheme and the time-multiplexing based scheme has been carried out. The optimal power allocation between pilot and data that maximizes a lower bound of  the maximum mutual information criterion has been provided.  It has been shown that  the use of  the optimal conventional superimposed training scheme entails a gain in terms of channel capacity only in special scenarios (many receive antennas and/or short coherence time). In other scenarios, the superimposed training scheme suffers from high channel estimation errors and its gain over the time-multiplexing based scheme is often lost. For this reason, many alternatives to the conventional superimposed training scheme have been proposed in recent works. 

In \cite{ghogho}, M. Ghogho et al proposed to introduce a distortion to the data symbols, prior to adding the known pilot in such a way to guarantee the orthogonality between pilot and data sequences. It is shown that the channel estimation performance is by far enhanced as compared to the standard superimposed scheme. This technique is referred to as the data-dependent superimposed training (DDST). While the DDST scheme exhibits the same channel performance as its TDMT counterpart, the effect of the introduced distortion may considerably affect the detection performance. The aim of this paper is thus to study the BER performance of the DDST and TDMT schemes and to evaluate to which extent, the performance of the DDST scheme is altered.  

In the literature, the few works focusing on BER performance have been based on unrealistic assumptions like the uncorrelation between the noise and channel estimation error, \cite{ieeeWC.wang07}, \cite{ieeeWC.au07}. These assumptions make calculations feasible for fixed size dimensions but are far away from being realistic. To make derivations possible while keeping realistic conditions, we will relax the assumption of finite size dimensions by allowing the space and time dimensions to grow to infinity at the same rate. Working with the asymptotic regime, allows us to simplify the
derivations and at the same time, we observe that the obtained results apply as well to usual
sample and antenna-array sizes. We show also that the obtained expressions can be used to determine the optimal power allocation that minimize the BER. 

The remainder of this paper is as follows: In the next section, we introduce the system model. After that, we review in section \ref{sec:channel_estimation} the channel estimation and data detection processes  for the TDMT and DDST schemes. Section \ref{sec:approximation} is dedicated to the derivation of the asymptotic BER expressions. 
Based on these results, we determine the optimal allocation of power between data and training for both schemes. 
Finally, simulation results are provided in section \ref{sec:simulation} to validate the analytical derivation.

{\bf Notation:} Subscripts $^{\tiny \mbox{H}}$, $^{\#}$ and $\mathrm{Tr}\left(.\right)$ denote hermitian, pseudo-inverse and trace operators. The statistical expectation and the Kronecker product are denoted by $\mathbb{E}$ and $\otimes$. The $(K\times K)$ identity matrix is denoted by $ {\bf I}_K$, and the $(Q\times Q)$ matrix of all ones by $ {\bf 1}_Q$. The $(i,j)$$^{\mathrm{th}}$ entry of a matrix $\bf A$ is denoted by ${\bf A}_{i,j}$.

\section{System model and problem setting}
\subsection{Time-division multiplexing scheme}
We consider a $M\times K$ MIMO system operating over a flat fading channel. Two phases are considered:\\
{\bf First phase:} In the first phase, each transmitting antenna sends $N_1$ pilot symbols. The received symbol $\Y_1$ writes as:
$$
\Y_1=\H\P_t+\V_1,
$$
where $\P_t$ is the $K\times N_1$ pilot matrix and 
\begin{assumption}
$\H$ is the $M\times K$ channel matrix with independent and identically distributed (i.i.d.) Gaussian variables with zero mean and variance $\frac{1}{K}$,
\label{ass:A1}
\end{assumption}
\begin{assumption}
$\V_1$ is the $M\times N_1$ matrix whose entries are i.i.d. with variance $\sigma_v^2$
\label{ass:A2}
\end{assumption}

{\bf Second phase}: In the second phase, $N_2$ data symbols with power $\sigma_{w_t}^2$ are sent by each antenna so that the received signal matrix $\Y_2$ writes as:
$$
\Y_2=\H\W_t+\V_2
$$
where
\begin{assumption}
$\W_t$ is the $K\times N_2$ data matrix with i.i.d. bounded data symbols of power $\sigma_{w_t}^2$ and $\V_2$ is the $M\times N_2$ additive Gaussian noise matrix with entries of zero mean and variance $\sigma_v^2$. Moreover,  $\W_t$ is independent of $\V_1$ and $\V_2$.
\label{ass:A3}
\end{assumption}
\subsection{Data-dependent superimposed training scheme (DDST)}
An other alternative to TDMT based schemes is to send the training and data sequences at the same time. Since data is transmitted all the times, these schemes allow efficient bandwidth efficiency but may suffer from the interference caused by the training sequence. Ghogho et al proposed thus to distort the data so that is becomes orthogonal to the training sequence. The proposed distortion matrix $\D$ is defined as:
$$
\D=\I_N-\J
$$
where $\J=\frac{K}{N}{\bf 1}_{\frac{N}{K}}\otimes \I_K$, (we assume that $\frac{N}{K}$ is integer valued, $N$ being the sample size). This distortion matrix was shown to be optimal in the sense that it minimizes the averaged euclidean distance between the distorted and non-distorted data, \cite{ghoghossp}. The received signal matrix at each block is therefore given by:
$$
\Y=\H\W_d(\I_N-\J)+\H\P_d+\V
$$
where 
\begin{assumption}
$\W_d$ is the data matrix with i.i.d. bounded data symbols of power $\sigma_{w_d}^2$, and $\V$ is the $M\times N$ matrix whose entries are i.i.d. zero mean with variance $\sigma_v^2$.
\label{ass:A4}
\end{assumption}
Moreover,  $\P_d$ is the $K\times N$ training matrix . The chosen pilot matrix $\P_d$ should fulfill two requirements. It should be orthogonal to the distortion matrix $\D$, thus satisfying $\D\transconj{\P}_d={\bf 0}$, and also verify the orthogonality relation $\P_d\transconj{\P}_d=N\sigma_{P_d}^2\I_K$ in order to minimize the channel estimation error subject to a fixed training power. A possible pilot matrix that meets these requirements is :
\begin{assumption}
\begin{equation}
\P_d(k,n)=\sigma_{P_d}\exp\left(\jmath 2\pi kn/K\right) \hspace{0.2cm} \textnormal{with}\hspace{0.2cm} k=0,\cdots,K-1 \hspace{0.2cm}\textnormal{and} \hspace{0.2cm} n=0,\cdots,N-1.
\label{eq:pilot}
\end{equation}
\label{ass:A5}
\end{assumption}
\section{Channel estimation and data detection}
\label{sec:channel_estimation}
\subsection{TDMT scheme}
In the first phase, we assume that the receiver estimates the channel in the least-square sense. Hence, the channel estimate is given by:
\begin{align*}
\widehat{\H}_t&=\Y_1\transconj{\P}_t\left(\P_t\transconj{\P}_t\right)^{-1}\\
&=\H+\V_1\transconj{\P}_t\left(\P_t\transconj{\P}_t\right)^{-1}\\
&=\H+\Delta\H_t
\end{align*}
where $\Delta\H_t=\V_1\transconj{\P}_t\left(\P_t\transconj{\P}_t\right)^{-1}$. 
Thus the mean square error writes as:
$$
{\rm MSE}_t=M\sigma_v^2\tr\left(\P_t\transconj{\P}_t\right)^{-1}
$$
As it has been shown in \cite{hassibi}, the optimal training matrix that minimizes the MSE under a constant training energy $N_1\sigma_{P_t}^2$ should satisfy:
\begin{assumption}
$$
\P_t\transconj{\P}_t=N_1\sigma_{P_t}^2\I_K
$$
\label{ass:A5bis}
\end{assumption}
where $\sigma_{P_t}^2$ denotes the amount of power devoted to the transmission of a pilot symbol. The optimal minimum value for the ${\rm MSE}_t$ is then given by:
$$
{\rm MSE}_t=\frac{KM\sigma_v^2}{N_1\sigma_{P_t}^2}.
$$
In the data transmission phase, the linear receiver uses the channel estimate in order to retrieve the transmitted data. After channel inversion, the estimated data matrix is given by:
$$
\widehat{\bf W}_t=\left(\widehat{\H}_t\right)^{\#}\Y_2
$$
where $\left(\widehat{\H}_t\right)^{\#}$ denotes the pseudo-inverse matrix of $\widehat{\H}_t$. Assuming that the channel estimation error is small, the pseudo-inverse of the estimated matrix can be approximated by the linear part of the Taylor expansion as, \cite{book.magnus07}:
\begin{equation}
\left(\widehat{\H}_t\right)^{\#}=\H^{\#}-\H^{\#}\Delta\H_t\H^{\#}+\H^{\#}\transconj{\left(\H^{\#}\right)}\Delta\H_t\left(\I_M-\H\H^{\#}\right)
\label{eq:inverse}
\end{equation}
Substituting $\H^{\#}$ by $\left(\transconj{\bf H}{\H}\right)^{-1}\transconj{\H}$ in \eqref{eq:inverse}, we obtain:
$$
\left(\widehat{\H}_t\right)^{\#}=\H^{\#}-\H^{\#}\Delta\H_t\H^{\#}+\left(\transconj{\bf H}\H\right)^{-1}\Delta\transconj{\H}_t\boldsymbol{\Pi}
$$
where $\boldsymbol{\Pi}=\I_M-\H\left(\transconj{\H}\H\right)^{-1}\transconj{\H}$ is the orthogonal projector on the null space of $\H$. Hence,  the zero-forcing estimate of the transmitted matrix can be expressed as:
$$
\widehat{\bf W}_t={\bf W}_t-{\bf H}^{\#}\Delta\H_t{\bf W}_t+\left({\bf H}^{\#}-{\bf H}^{\#}\Delta\H_t{\bf H}^{\#}\right){\bf V}_2+\left(\transconj{\bf H}\H\right)^{-1}\transconj{\left(\Delta\H_t\right)}{\boldsymbol{\Pi}}{\V}_2.
$$
Consequently, the effective post-processing noise $\Delta{\bf W}_t=\widehat{\bf W}_t-{\bf W}_t$ could be written as:
$$
\Delta{\bf W}_t=-{\bf H}^{\#}\Delta{\bf H}_t{\bf W}_t+\left({\bf H}^{\#}-{\bf H}^{\#}\Delta{\H}_t{\bf H}^{\#}+\left(\transconj{\H}\H\right)^{-1}\transconj{\left(\Delta\H_t\right)}\boldsymbol{\Pi}\right){\bf V}_2. 
$$
\subsection{DDST scheme}
The LS channel estimate is obtained by multiplying $\Y$ by $\transconj{\P}_d\left(\P_d\transconj{\bf P}_d\right)^{-1}$, thus giving:
$$
\widehat{\H}_d=\Y\transconj{\bf P}_d\left(\P_d\transconj{\bf P}_d\right)^{-1}=\H+\V\transconj{\bf P}_d\left(\P_d\transconj{\bf P}_d\right)^{-1}=\H+\Delta{\bf H}_d
$$
where $\Delta{\bf H}_d={\bf V}\transconj{\bf P}_d\left(\P_d\transconj{\bf P}_d\right)^{-1}$ denotes the channel estimation error matrix for the DDST scheme. 
As aforementioned above in assumption \ref{ass:A5}, the optimal training matrix that minimizes the MSE should satisfy:
$$
{\bf P}_d\transconj{\bf P}_d=N\sigma_{P_d}^2\I_K.
$$
The MSE is thus given by:
$$
{\rm MSE}_d=M\sigma_v^2\tr\left({\bf P}_d\transconj{\bf P}_d\right)^{-1}=\frac{KM\sigma_v^2}{N\sigma_{P_d}^2}
$$
For the DDST scheme, we consider a zero-forcing receiver which, prior to inverting the channel matrix, cancels the contribution of the training symbols by right multiplying $\Y$ by $\left(\I-\J\right)$, where
$$
{\bf Y}={\bf H}{\bf W}_d\left(\I_N-{\bf J}\right),
$$
the matrix ${\bf W}_d$ being the sent data matrix.
Thus, the zero-forcing estimate of ${\bf W}_d$ is given by:
\begin{align*}
\widehat{\bf W}_d&=\left(\widehat{\bf H}_d\right)^{\#}{\Y}\left(\I-\J\right)\\
&=\left({\bf H}^{\#}-{\bf H}^{\#}\Delta{\bf H}_d{\bf H}^{\#}+\left(\transconj{\bf H}\H\right)^{-1}\transconj{\Delta{\bf H}}_d\boldsymbol{\Pi}\right)\left({\bf HW}_d\left(\I-\J\right)+{\bf V}\left({\bf I}-{\bf J}\right)\right)\\
&=\left({\bf I}-{\bf H}^{\#}\Delta{\bf H}_d\right){\bf W}_d\left({\I}-\J\right)+\left({\bf H}^{\#}-{\bf H}^{\#}\Delta{\bf H}_d{\bf H}^{\#}\right)\V\left(\I-\J\right)+\left(\transconj{\bf H}\H\right)^{-1}\transconj{\Delta\H}_d\boldsymbol{\Pi}{\V}\left(\I-\J\right)\\
&=\W_d\left(\I-\J\right)-{\bf H}^{\#}\Delta\H_d{\W}_d\left(\I-\J\right)+\left(\H^{\#}-\H^{\#}\Delta\H_d\H^{\#}\right)\V\left(\I-\J\right)+\left(\transconj{\bf H}\H\right)^{-1}\transconj{\Delta\H}_d\boldsymbol{\Pi}\V\left(\I-\J\right)\\
&=\W_d+\left(-\W_d\J-{\bf H}^{\#}\Delta\H_d\W\left(\I-\J\right)+\left(\H^{\#}-\H^{\#}\Delta\H_d\H^{\#}\right){\V}(\I-\J)\right)+\left(\transconj{\H}\H\right)^{-1}\transconj{\Delta\H}_d\boldsymbol{\Pi}{\V}\left(\I-\J\right)
\end{align*}
Hence:
$$
\Delta{\bf W}_d=-{\bf W}_d{\bf J}-{\bf H}^{\#}\Delta{\bf H}_d\W_d\left(\I-\J\right)+\left({\bf H}^{\#}-{\bf H}^{\#}\Delta{\bf H}_d{\bf H}^{\#}\right){\bf V}\left(\I-\J\right)+\left(\transconj{\H}\H\right)^{-1}\transconj{\Delta\H}_d\boldsymbol{\Pi}{\V}\left(\I-\J\right).
$$
\section{Bit error rate performance}
\label{sec:approximation}
\subsection{TDMT scheme}
In order to evaluate the bit error rate performance, we need to evaluate the asymptotic behaviour of the post-processing noise observed at each entry of matrix $\Delta{\bf W}_t$. Using the 'characteristic function' approach, we can prove that conditioned on the channel matrix, the noise behaves asymptotically like a Gaussian random variable. This result is stated in the following theorem but its proof is postponed in appendix \ref{appendix:TDMT}.
\begin{theorem}
Under assumptions \ref{ass:A1}, \ref{ass:A2}, \ref{ass:A3}, \ref{ass:A5bis} and under the asymptotic regime defined as:
$$
M,K,N_1,N_2\to +\infty \hspace{0.1cm}\textnormal{with}\hspace{0.1cm} \frac{K}{N_1+N_2}\to c_1, 0<c_1<1 \hspace{0.1cm}\frac{M}{K}\to c_2>1 \hspace{0.1cm}\textnormal{and}\hspace{0.1cm}\frac{N_2}{N_1}\to r
$$
the post-processing noise experienced by the $i$-th antenna at each time $k$, $\Delta{\bf W}_t(i,k)$, for the TDMT scheme behaves in the asymptotic regime as a Gaussian random variable:\\
\begin{equation*}
\mathbb{E}\left[e^{\jmath\Re\left( z^*\Delta{\bf W}_t(i,k)\right)}\right]-e^{-\frac{\sigma_{w_t}^2\delta_t\left[\left({\bf H}^{\mbox{\tiny H}}{\bf H}\right)^{-1}\right]_{i,i}|z|^2}{4}}\xrightarrow[K\to+\infty]{} 0
\end{equation*}
where 
$$
\delta_t=c_1(1+r)\frac{\sigma_v^2}{\sigma_{P_t}^2}+\frac{\sigma_v^2}{\sigma_{w_t}^2}+\frac{c_1(1+r)(c_2+1)\sigma_v^4}{\sigma_{w_t}^2\sigma_{P_t}^2(c_2-1)}.
$$
\label{th:TDMT}
and  $K\to +\infty$ refers to this asymptotic regime.
\end{theorem}
\begin{remark}
	Note that as compared to the results in \cite{spawc2006}, our results make appear a new additive term of order $\sigma_v^4$. 
\end{remark}
The gaussianity of the post-processing noise being verified in the asymptotic case, we can derive the bit error rate for QPSK constellation and Gray encoding as \cite{book.proakis95}:
\begin{equation}
\mathrm{BER}=\mathbb{E}Q(\sqrt{x})
\label{eq:ber}
\end{equation}
where the expectation is taken with respect to the probability density function of the post processing SNR at the $i$-th branch  defined as:
$$
\gamma_t=\frac{1}{\delta_t\left[\left(\transconj{\bf H}{\bf H}\right)^{-1}\right]_{i,i}}.
$$
From \cite{gore} and \cite{winters}, we know that $\frac{1}{\left[\left(\transconj{\bf H}{\bf H}\right)^{-1}\right]_{i,i}}$ is a weighted chi-square distributed random variable with $2(M-K+1)$ degrees of freedom, whose density function is given by:
$$
f(x)=\frac{K^{M-K+1}x^{M-K}e^{-Kx}}{(M-K)!}{\bf 1}_{\left[0,+\infty\right[},
$$
where ${\bf 1}_{\left[0,+\infty\right[}$ is the indicator function corresponding to the interval $\left[0,+\infty\right[$.
Hence, the probability density function of $\gamma_t$ is given by:
\begin{equation}
f_{\gamma_t}(x)=\frac{(K\delta_t)^{M-K+1}x^{M-K}\exp(-K\delta_t x)}{(M-K)!}{\bf 1}_{\left[0,+\infty\right[}
\label{eq:pdf}
\end{equation}
Plugging \eqref{eq:pdf} into \eqref{eq:ber}, we get:
\begin{equation}
\mathrm{BER}_t=\frac{(K\delta_t)^{M-K+1}}{(M-K)!}\int_{0}^{+\infty}x^{M-K}\exp(-K\delta_t x) Q(\sqrt{x}) dx
\label{eq:ber_n}
\end{equation}
To compute \eqref{eq:ber_n}, we use the following integral function:
\begin{equation}
J(m,a,b)=\frac{a^m}{\Gamma(m)}\int_0^{+\infty} \exp(-ax)x^{m-1}Q(\sqrt{bx}) dx.
\label{eq:integral}
\end{equation}
The BER is therefore equal to:
\begin{equation}
\mathrm{BER}=J(M-K+1,K\delta_t,1).
\label{eq:BER_J}
\end{equation}
The integral in \eqref{eq:integral} has been shown to have,  for $c>0$  the following closed-form expression, \cite{thomas}:
$$
J(m,a,b)=\frac{\sqrt{c/\pi}\Gamma(m+\frac{1}{2})}{2(1+c)^{m+\frac{1}{2}}\Gamma(m+1)}{}_2F_1(1,m+\frac 1 2;m+1;\frac{1}{1+c}), \quad c=\frac{b}{2a}
$$
where ${}_2F_1(p,q;n,z)$ is the Gauss hyper-geometric function \cite{book.gradshteyn07}.
If $c=0$ equivalently $b=0$, it is easy to note that $J(m,a,0)$ is equal to $\frac{1}{2}$.
When $m$ is restricted to positive integer values, the above equation can be further simplified to \cite{iscas.xu05}:
\begin{equation}
J(m,a,b)=\frac{1}{2}\left[1-\mu\sum_{k=0}^{m-1}{2k \choose k} \left(\frac{1-\mu^2}{4}\right)^k\right]
\label{eq:J}
\end{equation}
where $\mu=\sqrt{\frac{c}{1+c}}$.
Plugging \eqref{eq:J} into \eqref{eq:BER_J}, we get:
\begin{empheq}[box=\fbox]{align}
\mathrm{BER}_t=\frac{1}{2}\left[1-\mu_t\sum_{k=0}^{M-K}{2k \choose k}\left(\frac{1-\mu_t^2}{4}\right)^k\right]
\label{eq:BER_t}
\end{empheq}
where $\mu_t=\sqrt{\frac{1}{2K\delta_t+1}}$.
\subsection{DDST scheme}
Unlike the TDMT scheme, the asymptotic distribution of entries of the post-processing noise matrix is not Gaussian. Actually, we prove that:
\begin{theorem}
Under assumptions \ref{ass:A4}, \ref{ass:A5}, and under the asymptotic regime defined as:
$$
\frac{K}{N}\to c_1, 0<c_1<1 \hspace{0.1cm}\textnormal{with}\hspace{0.1cm}\frac{M}{K}\to c_2 >1
$$
the post-processing noise experienced by the $i$-th antenna at each time $k$ behaves asymptotically as a Gaussian mixture random variable, i.e,
\begin{align}
\mathbb{E}\left[\exp\left(\jmath\Re\left(z^*\left[\Delta{\bf W}_d\right]_{i,k}\right)\right)\right]-\sum_{i=1}^\mathcal{Q}p_i\exp\left(\jmath\Re\left(z^*\alpha_i\right)\right)\exp\left(-\frac{|z|^2\delta_d\sigma_{w_d}^2\left[\left(\transconj{\bf H}{\bf H}\right)^{-1}\right]_{i,i}}{4}\right)\xrightarrow[K\to\infty]{} 0
\end{align}
where :
\begin{align}
\delta_d=(1-c_1)\left(\frac{c_1\sigma_v^2}{\sigma_{P_d}^2}+\frac{\sigma_v^2}{\sigma_{w_d}^2}+\frac{c_1\sigma_v^4(c_2+1)}{(c_2-1)\sigma_{P_d}^2\sigma_{w_d}^2}\right)
\end{align}
\label{theorem:DDST}
and $\mathcal{Q}$ is the cardinal of the set of all possible values of $\left[\overline{\bf W}\right]_{i,k}=c_1\sum_{k=1}^{\frac{1}{c_1}}\left[{\bf W}_d\right]_{i,k}$, and $p_i$ is the probability that $\left[\overline{\bf W}\right]_{i,k}$ takes the value $\alpha_i$.

We can also prove that conditioning on the fact that $\left[{\bf W}\right]_{i,k}=\epsilon_1\sqrt{\frac{\sigma_{w_d}^2}{2}}+\jmath\epsilon_2\sqrt{\frac{\sigma_{w_d}^2}{2}}$ where $\epsilon_{1}=\pm1$ and $\epsilon_2=\pm1$ the post-processing noise satisfies:
\begin{small}
\begin{gather}
\mathbb{E}\left[\exp\left(\jmath\Re\left(z^*\left[\Delta {\bf W}_d\right]_{i,k}\right)\right)|\left[{\bf W}\right]_{i,k}=(\epsilon_1+\jmath\epsilon_2)\sqrt{\frac{\sigma_{w_d}^2}{2}}\right]-\sum_{i=1}^{\mathcal{Q}^{'}}p_i^{'}\exp\left(
\jmath\Re\left(z^*\left(-c_1\left(\epsilon_1+\jmath\epsilon_2\right)\sqrt{\frac{\sigma_{w_d}^2}{2}}+\alpha_i^{'}
\right)\right)\right)\\
\times\exp\left(-\frac{|z|^2\delta_d\sigma_{w_d}^2\left[\left(\transconj{\bf H}{\bf H}\right)^{-1}\right]_{i,i}}{4}\right)\xrightarrow[K\to\infty]{} 0
\end{gather}
\end{small}
where $\mathcal{Q}^{'}$ is the cardinal of the set of all possible values $\overline{\overline{W}}_i=c_1\sum_{l=1}^{\frac{1}{c_1}-1}\left[{\bf W}\right]_{i,l}$, and $p_i^{'}$ is the probability that $\overline{\overline{W}}_i$ takes the value $\alpha_i^{'}$.
\end{theorem}
\begin{proof}
See Appendix \ref{appendix:DDST}.
\end{proof}
The assumption of the gaussianity of the post processing noise has been always assumed. For time division multiplexed training, this assumption is well-founded, since the post-processing noise, converges to a Gaussian distribution in the asymptotic regime, (see theorem \ref{th:TDMT}). 

In the superimposed training case, the distortion caused by the presence of data symbols affects the distribution of the post-processing noise which becomes asymptotically Gaussian mixture distributed. To assess the system performance in this particular case, we will start from the elementary definition of the bit error rate. Let $\Delta{\bf W}_{i,k}$ denotes the post processing noise experienced at the $i$-th antenna at time $k$ (we omit the subscript $d$ for ease of notations).
As it has been previously shown, $\Delta{\bf W}_{i,k}$ behaves as a Gaussian mixture random variable.
Let $\sigma_{d}^2$ be the asymptotic variance of $\Delta{\bf W}_{i,k}$, i.e, $\sigma_{d}^2=\sigma_{w_d}^2\delta_d\left[\left(\transconj{\bf H}{\bf H}\right)^{-1}\right]_{i,i}$. 

Using the symmetry of the transmitted data, the BER expression at the $i$th branch, under QPSK constellation and for a given channel realization is given by:
\begin{align*}
\mathrm{BER}_i&=\frac{1}{2}\mathbb{P}\left[\Re\left(\widehat{\bf W}_{i,k}\right)> 0 | \Re\left({\bf W}_{i,k}\right)=-\sqrt{\frac{\sigma_{w}^2}{2}}\right]+\frac{1}{2}\mathbb{P}\left[\Re\left(\widehat{\bf W}_{i,k}\right)< 0 | \Re\left({\bf W}_{i,k}\right)=\sqrt{\frac{\sigma_{w}^2}{2}}\right]\\
&=\frac{1}{2}\mathbb{P}\left[\Re\left(\Delta{\bf W}_{i,k}\right)>\sqrt{\frac{\sigma_{w}^2}{2}}\right]+\frac{1}{2}\mathbb{P}\left[\Re\left(\Delta{\bf W}_{i,k}\right)<-\sqrt{\frac{\sigma_{w}^2}{2}}\right]
\end{align*}
In the asymptotic regime, $\left(\Delta{\bf W}_{i,k}\right)$ converges to a mixed Gaussian distribution with the probability density function:
$$
f(x)=\frac{1}{\sqrt{\pi\sigma_{d}^2}}\sum_{s=1}^{\sqrt{\mathcal{Q}^{'}}}p_s\exp(-\frac{(x+c_1\epsilon\sqrt{\frac{\sigma_{w_d}^2}{2}}-\Re(\alpha_s))^2}{\sigma_{d}^2})
$$
Hence, conditioned on the channel, the asymptotic bit error rate can be approximated by:
\begin{align*}
\mathrm{BER}_{i,d}&=\frac{1}{2}\frac{1}{\sqrt{\pi\sigma_{w_d}^2}}\int_{\sqrt{\frac{\sigma_{w_d}^2}{2}}}^{+\infty}\sum_{s=1}^{\sqrt{\mathcal{Q}^{'}}}p_s^{'}\exp\left(-\frac{(x-c_1\sqrt{\frac{\sigma_{w_d}^2}{2}}-\Re(\alpha_s))^2}{\sigma_{w_d}^2}\right)dx \\
&+\frac{1}{2}\frac{1}{\sqrt{\pi\sigma_{w_d}^2}}\int_{-\infty}^{-\sqrt{\frac{\sigma_{w_d}^2}{2}}}\sum_{s=1}^{\sqrt{\mathcal{Q}^{'}}}p_s^{'}\exp\left(-\frac{(x+c_1\sqrt{\frac{\sigma_{w_d}^2}{2}}-\Re(\alpha_s))^2}{\sigma_{w_d}^2}\right)dx
\end{align*}
Finally, the proposed approximation of the BER  can be obtained by averaging with respect to the channel realization ${\bf H}$, thus giving:
\begin{align*}
\mathrm{BER}_{d}&=\mathbb{E}\frac{1}{2}\sum_{s=1}^{\sqrt{\mathcal{Q}^{'}}}p_s^{'}Q\left(\sqrt{\frac{\sigma_{w}^2}{\sigma_{w_d}^2}}(1-c_1)-\frac{\Re(\alpha_s)}{\sqrt{\frac{\sigma_{w_d}^2}{2}}}\right)+\frac{1}{2}\sum_{s=1}^{\sqrt{\mathcal{Q}^{'}}}p_s^{'}Q\left(\sqrt{\frac{\sigma_{w_d}^2}{\sigma_{w_d}^2}}(1-c_1)+\frac{\Re(\alpha_s)}{\sqrt{\frac{\sigma_{w_d}^2}{2}}}\right)\\
\end{align*}
For QPSK constellations, it can be shown that $\sqrt{\mathcal{Q}^{'}}=\frac{1}{c_1}$, where $\frac{1}{c_1}=\frac{N}{K}$ is assumed to be integer. Moreover, the set $\mathcal{S}$ of the values taken by $\Re(\alpha_s)$ can be given by:
$$
\mathcal{S}=\left\{\Re(\alpha_s)=c_1\sqrt{\frac{\sigma_{w_d}^2}{2}}(\frac{1}{c_1}-2s-1), s\in\left\{0,\cdots,\frac{1}{c_1}-1\right\}\right\}.
$$
with probability $p_s=\frac{{\frac{1}{c_1}-1 \choose s}}{2^{\frac{1}{c_1}-1}}$.

Let  $\gamma_d=\frac{\sigma_{w_d}^2}{\sigma_{d}}$ 
then, the BER expression becomes:
\begin{align}
\mathrm{BER}_{d}&=\mathbb{E}\sum_{s=0}^{\frac{1}{c_1}-1} \frac{{\frac{1}{c_1}-1\choose s}}{2^{\frac{1}{c_1}-1}}Q(2sc_1\sqrt{\gamma_d})
\label{eq:BER_n}
\end{align}
where the expectation is taken over the distribution of $\gamma_d$ given by:
$$
f_{\gamma_d}(x)=\frac{(K\delta_d)^{M-K+1}x^{M-K}}{(M-K)!}\exp(-K\delta_dx).
$$
The computation of the BER can be treated similarly to the TDMT scheme, thus leading to:
\begin{empheq}[box=\fbox]{align}
{\rm BER}_d&=\frac{1}{2^{\frac{1}{c_1}-1}}\sum_{s=0}^{\frac{1}{c_1}-1}{\frac{1}{c_1}-1 \choose s}J(M-K+1,K\delta_d,4s^2c_1^2)
\label{eq:BER_d}
\end{empheq}

\section{Optimal power allocation}
So far, we have provided approximations of the BER for the TDMT and DDST  schemes. As it has been previously shown, these expressions, depend on the power
allocated to data and training, in addition to other parameters. While the system has no control
over the noise power or the number of transmitting and receiving antennas, it still can optimize
the power allocation in such a way to minimize this performance index. Next, we provide, for the TDMT and DDST schemes, the optimal data and training power amounts that minimize the BER under the constraint of a constant total power.
\subsection{Optimal power allocation for the TDMT scheme}
Referring to the expressions of BER, we can easily see that the optimal amount of power allocated to data and pilot for the TDMT scheme is the one that minimizes $\delta_t$. Let $\tilde{c}_1=(1+r)c_1$, then 
minimizing $\delta_t$ with respect to $\sigma_{w_t}^2$ and $\sigma_{P_t}^2$ under the constraint that $N_1\sigma_{P_t}^2+N_2\sigma_{w_t}^2=(N_1+N_2)\sigma_T^2$ ($\sigma_T^2$ being the mean energy per symbol) results in the following lemma:
\begin{lemma}
The optimal power allocation minimizing the BER under
\begin{align}
\sigma_{w_t}^2&=\frac{(1+r)\sigma_T^2\sqrt{r\left((1+r)\sigma_T^2+\frac{\tilde{c}_1\sigma_v^2(c_2+1)}{c_2-1}\right)}}{r\left(\sqrt{r\left((1+r)\sigma_T^2+\frac{\tilde{c}_1\sigma_v^2(c_2+1)}{c_2-1}\right)}+\sqrt{\tilde{c}_1(\left((1+r)\sigma_T^2+\frac{r\sigma_v^2(c_2+1)}{c_2-1}\right)}\right)}\label{tmx_eq_power},\\
\sigma_{P_t}^2&=\frac{r(1+r)\sigma_T^2\sqrt{\tilde{c}_1\left((1+r)\sigma_T^2+\frac{r\sigma_v^2(c_2+1)}{c_2-1}\right)}}{r\left(\sqrt{r\left((1+r)\sigma_T^2+\frac{\tilde{c}_1\sigma_v^2(c_2+1)}{c_2-1}\right)}+\sqrt{\tilde{c}_1\left((1+r)\sigma_T^2+\frac{r\sigma_v^2(c_2+1)}{c_2-1}\right)}\right)}.\label{tmx_eq_pilot}
\end{align}
\end{lemma}
\subsection{Optimal power allocation for the DDST scheme}
For the DDST scheme, we can deduce from \eqref{eq:BER_n} that maximizing $\gamma_d$ leads to minimize the BER. To maximize $\gamma_d$, we need to optimize $\delta_d$ as a function of $\sigma_{w_d}^2$ and under the constraint that $\sigma_{P_d}^2+(1-c_1)\sigma_{w_d}^2=\sigma_T^2$. After straightforward calculations, we can find that the optimal values for $\sigma_{w_d}^2$ and $\sigma_{P_d}^2$ are given by:
\begin{lemma}
Under the data model, the optimal power allocation minimizing the BER under a total power constraint $\sigma_{T}^2$ is given by:
\begin{align}
\sigma_{w_d}^2&=\frac{\sqrt{(1-c_1)\left(\sigma_T^2+\frac{c_1(c_2+1)\sigma_v^2}{c_2-1}\right)}\sigma_T^2}{\left(1-c_1\right)\left(\sqrt{(1-c_1)\left(\sigma_T^2+\frac{c_1(c_2+1)\sigma_v^2}{c_2-1}\right)}+\sqrt{c_1\sigma_T^2+\frac{c_1(c_2+1)(1-c_1)\sigma_v^2}{c_2-1}}\right)},\label{eqpower}\\
\sigma_{P_d}^2&=\frac{\sqrt{c_1\sigma_T^2+\frac{c_1(c_2+1)(1-c_1)\sigma_v^2}{c_2-1}}\sigma_T^2}{\sqrt{(1-c_1)\left(\sigma_T^2+\frac{c_1(c_2+1)\sigma_v^2}{c_2-1}\right)}+\sqrt{c_1\sigma_T^2+\frac{c_1(c_2+1)(1-c_1)\sigma_v^2}{c_2-1}}}.\label{eqpilot}\\
\nonumber
\end{align}
\end{lemma}

\section{Discussion}
To get more insight into the proposed analysis, we provide here some comments and workouts on the theoretical results derived in the previous sections.\\[0.5cm]
\underline{\it High SNR  behaviour of the BER}: At high SNRs, the error variance parameters $\delta_t$ and $\delta_d$ are close to zero and hence, by using a first order Taylor expansion of the BER expressions in (\ref{eq:BER_t}) and (\ref{eq:BER_d}), we obtain:
\begin{eqnarray}
	BER_t &\approx& \frac{1}{2^{M-K+1}}(K\delta_t)^{M-K+1}{2(M-K)+1\choose M-K+1} \\
	BER_d &\approx&  \frac{1}{2^{\frac{1}{c_1}}} + O((K\delta_d)^{M-K+1})
\end{eqnarray}
where $O(x)$ denotes a real value of the same order of magnitude as $x$. From these approximated expressions, one can observe that the BER at the TDMT scheme is a monomial function of the estimation error variance parameter $\delta$ and the number of transmitters $K$. For example, if the noise power is decreased by a factor $2$, then the BER will decrease by  $2^{M-K+1}$. The diversity gain is thus equal to $M-K+1$, which is in accordance with the works in \cite{hedayat} and \cite{ieeeWC.wang07}.  
Also, we observe that, for the DDST case, we have a floor effect on the BER (i.e. the BER is lower bounded by $\frac{1}{2^{\frac{1}{c_1}}}$) due to the data distorsion inherent to this transmission scheme.
 \\[0.5cm]
\underline{\it Gaussian vs. Gaussian mixture model}: In our derivations we have found that the post-processing noise in the DDST case behaves asymptotically as a Gaussian mixture process while, in most of the existing works, the noise is assumed to be asymptotically Gaussian distributed. In fact, one can show that for large sample sizes (i.e. when $c_1 \longrightarrow 0$) the Gaussian mixture converges to a Gaussian distribution allowing us to retrieve the standard Gaussian noise assumption. However, for small or moderate sample sizes the considered Gaussian mixture model  leads to a much better approximation of the BER analytical expression than the one we would obtain with a post-processing Gaussian noise model. In other words, Theorem 2 results allow us to derive closed form expressions for the BER that are valid for relatively small sample sizes.\\[0.5cm]
\underline{\it Workouts on the optimal power allocation expressions of the TDMT scheme}: We consider here two limit cases: (i) The high SNR case where $\sigma_v^2 \ll \sigma_T^2$ and (ii) the case of high dimensional system (the number of transmit antennae is of the same order of magnitude as the number of receive antennae) where $c_2-1 \ll 1$. From (\ref{tmx_eq_pilot}), the data to pilot power ratio can then be approximated by:
\begin{eqnarray}
\mbox{case (i)} ~~~~ \frac{\sigma_{w_t}^2}{\sigma_{P_t}^2} &\approx& \frac{N_1}{\sqrt{N_2K}}  \label{workout_tdmt1}\\
\mbox{case (ii)} ~~~ \frac{\sigma_{w_t}^2}{\sigma_{P_t}^2} &\approx& \frac{N_1}{N_2}\label{workout_tdmt2}
\end{eqnarray}
Equation (\ref{workout_tdmt1}) shows that  the optimal power allocation in the high SNR case realizes  a kind of trade off between the pilot size and its power such that the total energy $N_1\sigma_{P_t}^2$ is kept constant. This suggests us to use the smallest possible pilot size that meet the technical constraint of limited transmit power,  to increase the effective channel throughput without loss of performance. \\
Equation (\ref{workout_tdmt1}) shows that in the difficult case of large dimensional system, one needs to allocate the same total energy to pilots and to data symbols, i.e. $N_1\sigma_{P_t}^2 \approx N_2\sigma_{w_t}^2$. In other word, we should give similar importance (in terms of power allocation) to the channel estimation  and to the data detection.
\\[0.5cm]
\underline{\it Workouts on the optimal power allocation expressions of the DDST scheme}: A similar workout is considered here for the DDST scheme. We consider the two previous limit cases and we assume that the sample size is much larger than the number of transmitters, i.e. $N \gg K$. In this context, we obtain the following approximations for the data to pilot power ratio:
 \begin{eqnarray}
\mbox{case (i)} ~~~~ \frac{\sigma_{w_d}^2}{\sigma_{P_d}^2} &\approx& \sqrt{\frac{N}{K}}  \label{workout_ddst1}\\
\mbox{case (ii)} ~~~ \frac{\sigma_{w_t}^2}{\sigma_{P_t}^2} &\approx& 1 \label{workout_ddst2}
\end{eqnarray}
Again we observe that for the large dimensional system case, one needs to allocate the same total energy to pilot and to the data. For high SNRs, one observe a kind of trade off between the pilot power and  size but in a different way than the TDMT case. In fact, if we increase by a factor of 4 the sample size, one can increase the  data to pilot power ratio by a factor of 2 without affecting the BER performance. 
\\[0.5cm]
\underline{\it High SNR BER comparison of the two pilot design schemes}: For the DDST scheme, the BER expression can be lower bounded as follows (using the convexity of $Q(\sqrt{bx})$ as a function of $b$):
\begin{eqnarray*}
{\rm BER}_d&=&\frac{1}{2^{\frac{1}{c_1}-1}}\sum_{s=0}^{\frac{1}{c_1}-1}{\frac{1}{c_1}-1 \choose s}J(M-K+1,K\delta_d,4s^2c_1^2) \\
&\geq& J(M-K+1,K\delta_d,\frac{1}{2^{\frac{1}{c_1}-1}}\sum_{s=0}^{\frac{1}{c_1}-1}{\frac{1}{c_1}-1 \choose s}4s^2c_1^2) = J(M-K+1,K\delta_d,1-c_1) \\
& \geq & J(M-K+1,K\delta_d,1)
\end{eqnarray*}
the latter inequality comes from the fact that  $J(m,a,b)$ is a decreasing function of its last argument. Now, in the high SNR and large sample size scenario (i.e, for $\sigma_v^2/\sigma_T^2 \ll 1$ and $N\gg N_1, K$), we have $\delta_t \approx \delta_d$ and by continuity  $ J(M-K+1,K\delta_d,1) \approx  J(M-K+1,K\delta_t,1)={\rm BER}_t$. Consequently, in this context, the TDMT scheme is better than the DDST in terms of BER, i.e. 
\[ {\rm BER}_d \geq {\rm BER}_t.\]

\section{Simulations}
\label{sec:simulation}
Despite being valid only for the asymptotic regime, our results are found to yield a good accuracy even for very small system dimensions. In this section, we present simulation results that compares between the TDMT and DDST schemes. 
\subsection{Performance comparison between DDST and TDMT based schemes}
In this section, except when mentioning, we consider a $2\times 4$ MIMO system ($K=2$, $M=4$) with a data block size $N=32$. 
\subsubsection{Bit error rate performance}
Fig. \ref{fig:BER_TDMT_DDST} plots the empirical and theoretical BER under QPSK constellation for $N=32$, $K=2$ and $M=4$ for the TDMT and DDST based schemes. All comparisons are conducted in the context when both schemes have the same total energy. The number of training symbols is set to $N_1=2$ ($N_2=30$) for the TDMT scheme.
\begin{figure}[htbp]
\centering 
\includegraphics[width=0.4\textwidth]{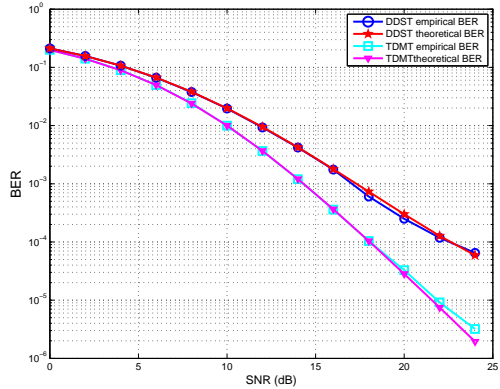}
\caption{Theoretical and empirical BER for the TDMT and DDST based schemes.}
\label{fig:BER_TDMT_DDST}
\end{figure}
For low SNR values (SNR below $6$ dB), both schemes achieve approximatively the same BER performance, and therefore, the DDST scheme outperforms its TDMT counterpart in terms of data rate, since it has a better bandwidth efficiency. For high SNR values, the noise caused by the data distortion is higher than the additive Gaussian noise, thus affecting the performance of the DDST scheme. 
\subsubsection{Applications}
To compare the efficiency of the TDMT and DDST schemes, we consider applications in which the BER should be below a certain threshold, say $10^{-2}$. This may be the case for instance of circuit-switched voice applications. Note that for non-coded systems, a target BER of $10^{-2}$ is commonly used.
\paragraph{Application 1}
In this scenario, we set the $\mathrm{SNR}\triangleq\frac{\sigma_T^2}{\sigma_v^2}$ to $15$ dB. We then vary the ratio $c_1=\frac{K}{N}$ from $0.01$ to $0.5$. Since we consider $K=2$ and $M=4$, $N=K/c_1$ varies also with $c_1$. For each value of $N$ we compute the BER by using \eqref{eq:BER_t} and \eqref{eq:BER_d}. Fig. \ref{fig:app} illustrates the obtained results. We also superposed in the same plot the empirical results for the TDMT and the DDST scheme. The results show a good match thereby supporting the usefulness of the derived results. 
\begin{figure}[htbp]
\centering 
\includegraphics[width=0.5\textwidth]{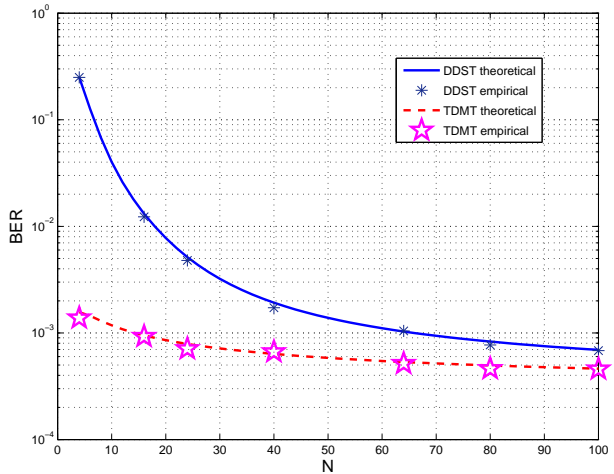}
\caption{BER with respect to $c_1$ when $K=2$, $M=4$ and SNR=15 dB}
\label{fig:app}
\end{figure}
We note that the DDST scheme may be interesting  for long enough frames ($N\geq 16$). For small frames (high distortion ratio $c_1$), the distortion of the data becomes too high thus reducing the interest of the DDST scheme. 
\paragraph{Application 2}
In this experiment, we propose to determine for the TDMT scheme ($K=2,M=4,N=32$) the optimal ratio $\frac{N_2}{N_1}$ that has to be used to meet a certain quality of service. For that, we  consider a scenario where the BER should be below $10^{-2}$. Using (\ref{tmx_eq_power}), (\ref{tmx_eq_pilot}) and (\ref{eq:BER_t}),  we determine the minimum number of required training symbols to meet the BER lower bound requirement. We then, plot the corresponding ratio $r=\frac{N_2}{N_1}$ with respect to the SNR.  We note that if the SNR is below $2$ dB, the BER requirement could not be achieved.
This is to be compared with the DDST scheme where the SNR should be set at least to $10.5$ dB so as to meet the BER lower bound requirement as it can be shown in fig. \ref{fig:r_value}.
Moreover, for a SNR more than $8.5$ dB, the minimum number of pilot symbols for channel identification  (equal to $K$) is sufficient to meet the BER requirement. 
\begin{figure}[h]
\centering 
\includegraphics[width=0.45\textwidth]{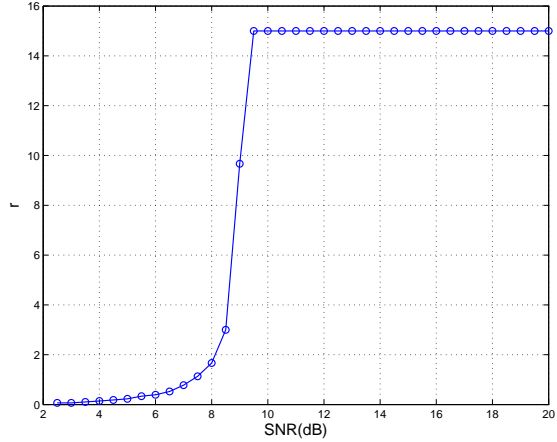}
\caption{Required $r$ versus SNR for $\mathrm{BER}\leq 10^{-2}$.}
\label{fig:r_value}
\end{figure}

\appendices
\section{Proof of theorem \ref{th:TDMT}}
\label{appendix:TDMT}
In the sequel, we propose to determine the asymptotic distribution of the post-processing noise of each entry of the matrix $\Delta {\bf W}_t$.
Actually the $(i,j)$ entry of $\Delta{\bf W}_t$ is given by:
$$
\left(\Delta{\bf W}_t\right)_{i,j}=-{\bf h}_i^{\#}\Delta{\bf H}_t{\bf w}_j+{\bf h}_i^{\#}\left({\bf I}_K-{\Delta{\bf H}_t}{\bf H}^{\#}\right){\bf v}_{2,j}+\widetilde{\bf h}_i\transconj{\left(\Delta{\bf H}_t\right)}\boldsymbol{\Pi}{\bf v}_{2,j}
$$
where ${\bf h}_i^{\#}$ and $\widetilde{\bf h}_i$ denote respectively the $i$th row of ${\bf H}^{\#}$ and $\left(\transconj{\bf H}{\bf H}\right)^{-1}$, and ${\bf w}_j$ and ${\bf v}_{2,j}$ denote $j$th columns of ${\bf W}_t$ and ${\bf V}_2$, respectively.
Conditioned on ${\bf H}$, ${{\bf V}_1}$ and ${\bf W}_t$, $\left(\Delta{\bf W}_t\right)_{i,j}$ is a Gaussian random variable with mean equal to $-{\bf h}_i^{\#}\Delta{\bf H}_t{\bf w}_j$ and variance 
\begin{eqnarray*}
\sigma_{w,K}^2&=&\sigma_v^2\left({\bf h}_i^{\#}-{\bf h}_i^{\#}\Delta{\bf H}_t{\bf H}^{\#}+\widetilde{\bf h}_i\transconj{\left(\Delta{\bf H}_t\right)}\boldsymbol{\Pi}\right)\left(\transconj{\left({\bf h}_i^{\#}\right)}-\transconj{\left({\bf H}^{\#}\right)}\transconj{\Delta{\bf H}}_t\transconj{\left({\bf h}_i^{\#}\right)}+\boldsymbol{\Pi}\Delta{\bf H}_t\transconj{\left(\widetilde{\bf h}_i\right)}\right)
\end{eqnarray*}
Since our proof will be based on the 'characteristic function' approach, we shall first recall the expression of the characteristic function for complex random variables:
\begin{theorem}
Let $X_n$ be a complex Gaussian random variable with mean $m_{X,n}$ and variance $\sigma_{X,n}^2$,  such that $\mathbb{E}(X_n-m_{X,n})^2=0$. Then, $X_n$ can be seen as a two-dimensional random variable corresponding to its real and imaginary parts. The characteristic function of $X_n$ is therefore given by:
$$
 \mathbb{E}\left[\exp\left(\jmath\Re(z^*X_n)\right)\right]=\exp\left(\jmath\Re\left(z^*m_{X,n}\right)\right)\exp\left(-\frac{1}{4}|z|^2\sigma_{X,n}^2\right).
$$
\label{th:characteristic}
\end{theorem}
Applying Theorem\ref{th:characteristic}, the conditional characteristic function of $\left(\Delta{\bf W}\right)_{i,j}$ can be written as:
\begin{equation}
\mathbb{E}\left[\exp\left(\jmath\Re\left(z^*\left(\Delta{\bf W}_t\right)_{i,j}\right)\right)|{\bf V}_1,{\bf H},{\bf W}_t\right]=\exp\left(-\jmath \Re\left(z^*{\bf h}_i^{\#}\Delta{\bf H}_t{\bf w}_j\right)\right)\exp\left(-\frac{1}{4}|z|^2\sigma_{w,K}^2\right).
\label{eq:characteristic}
\end{equation}
To remove the condition expectation on ${\bf V}_1$ and ${\bf W}_t$, one should prove that $\sigma_{w,K}^2$ converges almost surely to a deterministic quantity. Actually, $\sigma_{w,K}^2$ can be expanded as follows:
\begin{align*}
\sigma_{w,K}^2&=\sigma_v^2{\bf h}_i^{\#}\transconj{\left({\bf h}_i^{\#}\right)}+\sigma_v^2{\bf h}_i^{\#}\Delta{\bf H}_t\left(\transconj{\bf H}{\bf H}\right)^{-1}\transconj{\left(\Delta{\bf H}_t\right)}\transconj{\left({\bf h}_i^{\#}\right)}-2\sigma_v^2\Re\left({\bf h}_i^{\#}\Delta{\bf H}_t\left(\transconj{\bf H}{\bf H}\right)^{-1}\transconj{\left({\bf h}_i^{\#}\right)}\right)\\
&+\sigma_v^2\widetilde{\bf h}_i\transconj{\Delta{\bf H}}_t\boldsymbol{\Pi}\Delta{\bf H}_t\transconj{\left(\widetilde{\bf h}_i\right)}.
\end{align*}
Let 
\begin{align*}
A_{\sigma,K}&=\sigma_v^2{\bf h}_i^{\#}\Delta{\bf H}_t\left(\transconj{\bf H}{\bf H}\right)^{-1}\transconj{\left(\Delta{\bf H}_t\right)}\transconj{\left({\bf h}_i^{\#}\right)}\\
B_{\sigma,K}&=\sigma_v^2 \widetilde{\bf h}_i\transconj{\Delta{\bf H}}_t\boldsymbol{\Pi}\Delta{\bf H}_t\transconj{\left(\widetilde{\bf h}_i\right)}\\
{\epsilon}_{\sigma,K}&={\bf h}_i^{\#}\Delta{\bf H}_t\left(\transconj{\bf H}{\bf H}\right)^{-1}\transconj{\left({\bf h}_i^{\#}\right)}.
\end{align*}
The limiting behaviour of $A_{\sigma,K}$ can be derived by using the following known results describing the asymptotic behaviour of an important class of quadratic forms:
\begin{lemma}{\cite[Lemma 2.7]{anlpro.silverstein98}}
\label{lemma:silverstein}
Let ${\bf x}=\transpose{\left[X_1,\cdots,X_N\right]}$ be a $N\times 1$ vector where the $X_n$ are centered i.i.d. complex random variables with unit variance. Let ${\bf A}$ be a deterministic $N\times N$ complex matrix. Then, for any $p\geq2$ there exists a constant $C_p$ depending on $p$ only such that:
\begin{equation}
\mathbb{E}\left|\frac{1}{N}\transconj{\bf x}{\bf A}{\bf x}-\frac{1}{N}\mathrm{Tr}({\bf A})\right|^p\leq \frac{C_p}{N^p}\left(\left(\mathbb{E}|X_1|^4\mathrm{Tr}\left({\bf A}\transconj{\bf A}\right)\right)^{p/2}+\mathbb{E}|X_1|^{2p}\mathrm{Tr}\left(\left({\bf A}\transconj{\bf A}\right)^{p/2}\right)\right)
\label{eq:silverstein}
\end{equation}
Noticing that $\mathrm{Tr}\left({\bf A}\transconj{\bf A}\right)\leq N \|{\bf A}\|^2$ and that $ \mathrm{Tr}\left(\left({\bf A}\transconj{\bf A}\right)^{p/2}\right)\leq N \|{\bf A}\|^p$, we obtain the simpler inequality:
\begin{equation}
\mathbb{E}\left|\frac{1}{N}\transconj{\bf x}{\bf A}{\bf x}-\frac{1}{N}\mathrm{Tr}({\bf A})\right|^p\leq\frac{C_p}{N^{p/2}}\|{\bf A}\|^p\left(\left(\mathbb{E}|X_1|^2\right)^{p/2}+\mathbb{E}|X_1|^{2p}\right)
\label{eq:silverstein_simple}
\end{equation}
\end{lemma}
Hence, if ${\bf A}$ and ${\bf x}$ have respectively finite spectral norm and finite eigth moment, we can conclude, using Borel-Cantelli lemma, about the almost convergence of the quadratic form $\frac{1}{N}\transconj{\bf x}{\bf A}{ \bf x}$, thus yielding the following corollary:
\begin{corollary}
 Let ${\bf x}=\transpose{\left[x_1,\cdots,x_N\right]}$ be a $N\times 1$ vector where the entries $x_i$ are centered i.i.d. complex random variables with unit variance and finite eight order. Let ${\bf A}$ be a determinsitic $N\times N$ complex matrix with bounded spectral norm. Then,
$$
\frac{1}{N}\transconj{\bf x}{\bf A} {\bf x}-\frac{1}{N}\mathrm{Tr}({\bf A}) \longrightarrow 0 \hspace{1cm}\textnormal{almost surely.}
$$
\label{corollary:quadratic}
\end{corollary}
By corollary \ref{corollary:quadratic}, the asymptotic behavior of ${\bf A}_{\sigma,K}$ is then given by:
$$
{A}_{\sigma,K}- \frac{\sigma_v^2\left[\left(\transconj{\bf H}{\bf H}\right)^{-1}\right]_{i,i}}{N_1\sigma_{P}^2}\mathrm{Tr}\left(\transconj{\bf H}{\bf H}\right)^{-1}\longrightarrow 0 \hspace{1cm}\textnormal{almost surely.}
$$
Since $\frac{1}{K}\mathrm{Tr}\left(\transconj{\bf H}{\bf H}\right)^{-1}$ converges asymptotically to $\frac{1}{c_2-1}$ as the dimensions go to infinity \cite{book.verdu04}, we get:
$$A_{\sigma,K}-\frac{c_1(1+r)\sigma_v^4}{(c_2-1)\sigma_{P_t}^2}\left[\left(\transconj{\bf H}{\bf H}\right)^{-1}\right]_{i,i}\longrightarrow 0.
$$
Note that Theorem\ref{corollary:quadratic} can be applied since the smallest eigenvalue of the Wishart matrix  $\left(\transconj{\bf H}{\bf H}\right)$ are almost surely uniformely bounded away from zero  by $(1-\sqrt{c_2})^2> 0$,  \cite{anlpro.silverstein85}. 

Before determining the limiting behavior of $B_{\sigma,K}$, we shall need the following lemma:
\begin{lemma}
Let ${\bf Y}=\left(\frac{1}{\sqrt{K}} y_{i,j}\right)_{i=1,j=1}^{M,K}$ be a $M\times K$ with Gaussian i.i.d entries. Then, in the asymptotic regime given by:
$$
M,K\to \infty \hspace{0.2cm} \textnormal{such that} \hspace{0.2cm} \frac{M}{K}\to c_2>1 
$$
we have:
$$
\left[\left(\transconj{\bf Y}{\bf Y}\right)^{-2}\right]_{i,i}-\frac{c_2}{c_2-1}\left(\left[\left(\transconj{\bf Y}{\bf Y}\right)^{-1}\right]_{i,i}\right)^2\to 0
$$
\label{lemma:new}
\end{lemma}
\begin{proof}
Without loss of generality, we can restrict our proof to the case where $i=1$. Let ${\bf y}_1,\cdots,{\bf y}_K$ denote the columns of ${\bf Y}$. Matrix $\transconj{\bf Y}{\bf Y}$ is then given by:
$$
\transconj{\bf Y}{\bf Y}=\begin{bmatrix}
\transconj{\bf y}_1{\bf y}_1 & \transconj{\bf y}_1{\bf y}_2 & \cdots & \transconj{\bf y}_1{\bf y}_K \\
\vdots & & & \vdots \\
\transconj{\bf y}_K{\bf y}_1 & \transconj{\bf y}_K{\bf y}_2 & \cdots & \transconj{\bf y}_K{\bf y}_K
\end{bmatrix}
$$
Let ${\bf v}_y=\left[\left[\left(\transconj{\bf Y}{\bf Y}\right)^{-1}\right]_{1,2},\cdots,\left[\left(\transconj{\bf Y}{\bf Y}\right)^{-1}\right]_{1,K}\right]$. Then, using the formula of the inverse of block matrices, we get:
$$
{\bf v}_y=-\left[\left(\transconj{\bf Y}{\bf Y}\right)^{-1}\right]_{1,1}\transconj{\bf y}_1\widetilde{\bf Y}\left(\transconj{\widetilde{\bf Y}}\widetilde{\bf Y}\right)^{-1}
$$
where $\widetilde{\bf Y}=\left[{\bf y}_2,\cdots,{\bf y}_K\right]$.

On the other hand,
\begin{align*}
\left[\left(\transconj{\bf Y}{\bf Y}\right)^{-2}\right]_{1,1}&=\left(\left[\left(\transconj{\bf Y}{\bf Y}\right)^{-1}\right]_{1,1}\right)^2+{\bf v}_y\transconj{\bf v}_y\\
&=\left(\left[\left(\transconj{\bf Y}{\bf Y}\right)^{-1}\right]_{1,1}\right)^2\left(1+\transconj{\bf y}_1\widetilde{\bf Y}\left(\transconj{\widetilde{\bf Y}}\widetilde{\bf Y}\right)^{-2}\transconj{\widetilde{\bf Y}}{\bf y}_1\right)
\end{align*}
Using corollary \ref{corollary:quadratic}, we have:
$$
\transconj{\bf y}_1\widetilde{\bf Y}\left(\transconj{\widetilde{\bf Y}}\widetilde{\bf Y}\right)^{-2}\transconj{\widetilde{\bf Y}}{\bf y}_1-\frac{1}{K}\mathrm{Tr}\left(\transconj{\widetilde{\bf Y}}\widetilde{\bf Y}\right)^{-1}\to 0 \hspace{0.2cm}\textnormal{almost surely.}
$$
Since $\frac{1}{K}\mathrm{Tr}\left(\transconj{\widetilde{\bf Y}}\widetilde{\bf Y}\right)^{-1}$ tends to $\frac{1}{c_2-1}$ almost surely, we get the desired result.
\end{proof}
We are now in position to deal with the term ${B}_{\sigma,K}$. Using corollary \ref{corollary:quadratic}, we get:
$$
{B}_{\sigma,K}-\frac{\sigma_v^4(M-K)}{N_1\sigma_{P}^2}\left[\left(\transconj{\bf H}{\bf H}\right)^{-2}\right]_{i,i} \to 0 \hspace{0.2cm}\textnormal{almost surely}
$$
Hence, 
$$
{B}_{\sigma,K}-\frac{\sigma_v^4 c_1(c_2-1)(1+r)}{\sigma_{P}^2}\left[\left(\transconj{\bf H}{\bf H}\right)^{-2}\right]_{i,i} \to 0 \hspace{0.2cm} \textnormal{almost surely}
$$
Using lemma \ref{lemma:new}, we get that:
$$
{B}_{\sigma,K}-\frac{\sigma_v^4 c_1c_2(1+r)}{\sigma_{P}^2}\left(\left[\left(\transconj{\bf H}{\bf H}\right)^{-1}\right]_{i,i}\right)^2 \to 0 \hspace{0.2cm} \textnormal{almost surely}
$$
It can be shown that $\left[\left(\transconj{\bf H}{\bf H}\right)^{-1}\right]_{i,i}$ converge almost surely to $\frac{1}{c_2-1}$, (its inverse is  the mean of independent random variables \cite{winters} ), then:
$$
{B}_{\sigma,K}-\frac{\sigma_v^4 c_1c_2(1+r)}{\sigma_{P}^2(c_2-1)}\left[\left(\transconj{\bf H}{\bf H}\right)^{-1}\right]_{i,i} \to 0 \hspace{0.2cm} \textnormal{almost surely}
$$
To prove the almost sure convergence to zero of $\epsilon_{\sigma,K}$, we will be based on the following result, about the asymptotic behaviour of weighted averages:
\begin{theorem}{Almost sure convergence of weighted averages \cite{jtp.baxter04}}
Let ${\bf a}=\transpose{\left[a_1,\cdots,a_N\right]}$ be a sequence of $N\times 1$ deterministic real vectors with $\sup_N\frac{1}{N}\transpose{\bf a}_N{\bf a}_N <+\infty$. Let ${\bf x}_N=\left[x_1,\cdots,x_N\right]$ be a $N\times 1$ real random vector with i.i.d. entries, such that $\mathbb{E}x_1=0$ and $\mathbb{E}|x_1| < +\infty$. Therefore, $\frac{1}{N}\transpose{\bf a}_N {\bf x}_N$ converges almost surely to zero as $N$ tends to infinity. 
\label{th:almostsure}
\end{theorem}
This theorem was proved in \cite{jtp.baxter04} for real variables. Since we are interested in the asymptotic convergence of the real part of $\epsilon_{\sigma,K}$, it can be possible to transpose our problem into the real case. Indeed,
let ${\bf x}=\transconj{\bf V}_1{\bf h}_i^{\#}$ and ${\bf a}=\transconj{\bf P}_t\left(\transconj{\bf H}{\bf H}\right)^{-1}{\bf h}_i^{\#}$, then $\Re\left(\epsilon_{\sigma,K}\right)$ is given by:
$$
\Re\left(\epsilon_{\sigma,K}\right)=\frac{1}{N_1\sigma_{P}^2}\Re(\transconj{\bf x}{\bf a})
$$
Let ${\bf a}_r, {\bf x}_r$ (resp. ${\bf a}_i, {\bf x}_i$) denote respectively the real parts (resp. imaginary parts) of ${\bf a}$ and ${\bf x}$, then 
$$
\Re\left(\epsilon_{\sigma,K}\right)=\frac{1}{N_1\sigma_{P}^2}\transpose{\bf a}_r{\bf x}_r-\transpose{\bf a}_i{\bf x}_i
$$
Referring to theorem \ref{th:almostsure}, the convergence to zero of $\Re\left(\epsilon_{\sigma,K}\right)$ is ensured if  $\frac{1}{2N_1}\left(\transpose{\bf a}_r{\bf a}_r+\transpose{\bf a}_i{\bf a}_i\right)=\frac{1}{2N_1}\|{\bf a}\|_2^2$is finite. This is almost surely true, since:
\begin{align*}
\frac{1}{N_1\sigma_{P}^2}\|{\bf a}\|_2^2&=\frac{1}{N_1\sigma_{P}^2}\mathrm{Tr}\left(\transconj{\bf P}_t\left(\transconj{\bf H}{\bf H}\right)^{-1}{\bf h}_i^{\#}\transconj{\left({\bf h}_i^{\#}\right)}\left(\transconj{\bf H}{\bf H}\right)^{-1}{\bf h}_i^{\#}\right)\\
&={\bf h}_i^{\#}\left(\transconj{\bf H}{\bf H}\right)^{-2}\transconj{\left({\bf h}_i^{\#}\right)} <\|\left(\transconj{\bf H}{\bf H}\right)^{-2}\|_2\left[\left(\transconj{\bf H}{\bf H}\right)^{-1}\right]_{i,i}\\
\end{align*}
This leads to 
$$
\sigma_{w,K}^2-\tilde{\sigma}_{w,K}^2 \longrightarrow 0 \hspace{1cm}\textnormal{almost surely}.
$$
where $\tilde{\sigma}_{w,K}^2$ is given by:
$$
\tilde{\sigma}_{w,K}^2=\sigma_v^2\left[\left(\transconj{\bf H}{\bf H}\right)^{-1}\right]_{i,i}+\frac{c_1(c_2+1)(1+r)\sigma_v^4}{(c_2-1)\sigma_{P}^2}\left[\left(\transconj{\bf H}{\bf H}\right)^{-1}\right]_{i,i}.
$$
Substituting $\sigma_{w,K}^2$ by its asymptotic equivalent in (\ref{eq:characteristic}), we get:
\begin{small}
$$
\mathbb{E}\left[\exp\left(\jmath\Re\left(z^*\left(\Delta{\bf W}_t\right)_{i,j}\right)\right)|{\bf H},{\bf W}_t\right]-\mathbb{E}\left[\exp\left(-\jmath\Re\left(z^*{\bf h}_i^{\#}\Delta{\bf H}_t{\bf w}_j\right)\right)|{\bf W},{\bf H}\right]\exp\left(-\frac{1}{4}|z|^2\tilde{\sigma}_{w,K}^2\right)\longrightarrow 0 \hspace{0.2cm} \textnormal{almost surely}.
$$
\end{small}
Also conditioning on ${\bf W}_t$ and ${\bf H}$, ${\bf h}_i^{\#} \Delta {\bf H}_t {\bf w}_j$ is a Gaussian random variable with zero mean and variance $$\sigma_{m,K}^2=\frac{\sigma_v^2}{N_1\sigma_{P}^2}{\bf h}_i^{\#}\transconj{{\bf w}_j}{\bf w}_j\left({\bf h}_i\right)^{\#}.$$ 

Since $\frac{1}{K}\transconj{{\bf w}_j}{\bf w}_j\longrightarrow \sigma_{w_t}^2$ almost surely, we get that $\sigma_{m,K}^2$ converges almost surely to $\tilde{\sigma}_{m,K}^2$ where
$$
\tilde{\sigma}_{m,K}^2=\frac{c_1(1+r)\sigma_v^2\sigma_{w_t}^2}{\sigma_{P_t}^2}\left[\left(\transconj{\bf H}{\bf H}\right)^{-1}\right]_{i,i},
$$
Using the fact that the characteristic function of ${\bf h}_i^{\#}\Delta{\bf H}_t{\bf w}_j$ is 
$$
\mathbb{E}\left[\exp\left(-\jmath\Re\left(z^*{\bf h}_i^{\#}\Delta{\bf H}_t{\bf w}_j\right)\right)|{\bf W},{\bf H}\right]=\exp\left(-\frac{1}{4}|z|^2\sigma_{m,K}^2\right),
$$
 we obtain that conditionally on the channel:
$$\mathbb{E}\left[\exp\left(\jmath\Re\left(z^*\left(\Delta{\bf W}_t\right)_{i,j}\right)\right)\right]-\exp\left(-\frac{1}{4}|z|^2\left(\tilde{\sigma}_{m,K}^2+\tilde{\sigma}_{w,K}^2\right)\right)\longrightarrow 0 \hspace{0.2cm}\textnormal{almost surely.}
$$
We end up the proof  by noticing that  $\tilde{\sigma}_{m,K}^2+\tilde{\sigma}_{w,K}^2=\sigma_{w_t}^2\delta_t\left[\left(\transconj{\bf H} {\bf H}\right)^{-1}\right]_{i,i}$.
\section{Proof of theorem \ref{theorem:DDST}}
\label{appendix:DDST}
For the DDST scheme, the post-processing noise matrix $\Delta {\bf W}_d$ is given by:
\begin{align*}
\Delta {\bf W}_d&=-{\bf WJ}-{\bf H}^{\#}{\Delta}{\bf H}_d{\bf W}\left({\bf I}_N-{\bf J}\right)+\left({\bf H}^{\#}-{\bf H}^{\#}{\Delta}{\bf H}_d{\bf H}^{\#}\right){\bf V}\left({\bf I}_N-{\bf J}\right)\\
&+\left(\transconj{\bf H}{\bf H}\right)^{-1}\transconj{\Delta{\bf H}}_d\boldsymbol{\Pi}{\bf V}\left({\bf I}_N-{\bf J}\right)\\
&=-{\bf WJ}-{\bf H}^{\#}{\Delta}{\bf H}_d{\bf W}\left({\bf I}_N-{\bf J}\right)+{\bf H}^{\#}{\bf V}\left({\bf I}_N-{\bf J}\right)-{\bf H}^{\#}\Delta{\bf H}_d{\bf H}^{\#}{\bf V}\left({\bf I}_N-{\bf J}\right)\\
&+\left(\transconj{\bf H}{\bf H}\right)^{-1}\transconj{\Delta{\bf H}}_d\boldsymbol{\Pi}{\bf V}\left({\bf I}_N-{\bf J}\right).
\end{align*}
Hence,
\begin{small}
\begin{align*}
\left(\Delta{\bf W}_d\right)_{i,j}&=-\tilde{\bf w}_i{\bf J}_j-{\bf h}^{\#}_i{\bf V}\transconj{\bf P}\left({\bf P}\transconj{\bf P}\right)^{-1}{\bf W}\left({\bf e}_j-{\bf J}_j\right)+{\bf h}^{\#}_i{\bf V}\left({\bf e}_j-{\bf J}_j\right)-{\bf h}^{\#}_i{\bf V}\transconj{{\bf P}}\left({\bf P}\transconj{\bf P}\right)^{-1}{\bf H}^{\#}{\bf V}\left({\bf e}_j-{\bf J}_j\right)\\
&+\widetilde{\bf h}_i\left({\bf P}\transconj{\bf P}\right)^{-1}{\bf P}\transconj{\bf V}\boldsymbol{\Pi}{\bf V}\left({\bf e}_j-{\bf J}_j\right)
\end{align*}
\end{small}
where ${\bf e}_j$ and ${\bf J}_j$ denotes the $j$th columns of ${\bf I}_N$ and ${\bf J}$, respectively and  $\tilde{\bf w}_i$ denotes the $i$th row of the matrix ${\bf W}$.

Let ${\bf v}_1={\bf V}\left({\bf e}_j-{\bf J}_j\right)$, and ${\bf v}_2={\mathrm{vec}({\bf V}\left({\bf P}\transconj{\bf P}\right)^{-1}\transconj{\bf P})}$

The vector $\transpose{\left[\transpose{\bf v}_1,\transpose{\bf v}_2\right]}$ is a Gaussian vector. Since $\mathbb{E}\left[{\bf v}_1\transconj{\bf v}_2\right]=0$, we conclude that ${\bf v}_1$ and ${\bf v}_2$ are independent. Then ${\bf v}_1$ and ${\bf V}_2={\bf V}\left({\bf P}\transconj{\bf P}\right)^{-1}\transconj{\bf P}$ are also independent. 
Moreover, $\mathbb{E}\left[{\bf v}_1\transconj{\bf v}_1\right]=\sigma_v^2\left(1-\frac{K}{N}\right){\bf I}_N$.

Conditioning on ${\bf V}_2$, ${\bf H}$ and ${\bf W}$, $\left(\Delta{\bf W}_d\right)_{i,j}$ is a Gaussian random variable with mean equal to $-\tilde{\bf w}_i{\bf J}_j-{\bf h}_i^{\#}{\bf V}_2{\bf W}\left({\bf e}_j-{\bf J}_j\right)$ and variance $\sigma_{w_d,N}^2$ equal to:
\begin{small}
\begin{align*}
\sigma_{w_d,N}^2&=\mathbb{E}\left[\left({\bf h}_i^{\#}-{\bf h}_i^{\#}{\bf V}_2{\bf H}^{\#}+\widetilde{\bf h}_i\transconj{\bf V}_2\boldsymbol{\Pi}\right){\bf v}_1\transconj{\bf v}_1\left(\transconj{\left({\bf h}_i^{\#}\right)}-\transconj{\left({\bf H}^{\#}\right)}\transconj{\bf V}_2\transconj{\left({\bf h}_i^{\#}\right)}+\boldsymbol{\Pi}{\bf V}_2\transconj{\widetilde{\bf h}}_i
\right)|{\bf V}_2\right]\\
&=\mathbb{E}\left[{\bf h}_i^{\#}{\bf v}_1\transconj{\bf v}_1\transconj{\left({\bf h}_i^{\#}\right)}\right]+\mathbb{E}\left[{\bf h}_i^{\#}{\bf V}_2{\bf H}^{\#}{\bf v}_1\transconj{\bf v}_1\transconj{\left({\bf H}^{\#}\right)}\transconj{\bf V}_2\transconj{\left({\bf h}_i^{\#}\right)}\right]-2\mathbb{E}\left[\Re\left({\bf h}_i^{\#}{\bf V}_2{\bf H}^{\#}{\bf v}_1\transconj{\bf v}_1\transconj{\left({\bf h}_i^{\#}\right)}\right)\right]\\
&+\sigma_v^2(1-\frac{K}{N})\widetilde{\bf h}_i\transconj{\bf V}_2\boldsymbol{\Pi}{\bf V}_2\transconj{(\widetilde{\bf h}_i)}\\
&=(1-\frac{K}{N})\sigma_v^2\left[\left(\transconj{\bf H}{\bf  H}\right)^{-1}\right]_{i,i}+\sigma_v^2(1-\frac{K}{N}){\bf h}_i^{\#}{\bf V}_2\left(\transconj{\bf H}{\bf H}\right)^{-1}\transconj{\bf V}_2\transconj{\left({\bf h}_i^{\#}\right)}
-2(1-\frac{K}{N})\sigma_v^2\Re\left({\bf h}_i^{\#}{\bf V}_2{\bf H}^{\#}\transconj{\left({\bf h}_i^{\#}\right)}\right)\\
&+\sigma_v^2(1-\frac{K}{N})\widetilde{\bf h}_i\transconj{\bf V}_2\boldsymbol{\Pi}{\bf V}_2\transconj{(\widetilde{\bf h}_i)}
\end{align*}
\end{small}
Using the same techniques as before, it can be proved that:
$$
(1-\frac{K}{N})\sigma_v^2{\bf h}_i^{\#}{\bf V}_2\left(\transconj{\bf H}{\bf H}\right)^{-1}\transconj{\bf V}_2\transconj{\left({\bf h}_i^{\#}\right)}-\frac{c_1(1-c_1)\sigma_v^4}{(c_2-1)\sigma_P^2}\left[\left(\transconj{\bf H}{\bf H}\right)^{-1}\right]_{i,i}\to 0 \hspace{0.1cm}\textnormal{almost surely}.
$$
and also that,
$$
\Re\left({\bf h}_i^{\#}{\bf V}_2{\bf H}^{\#}\transconj{\left({\bf h}_i^{\#}\right)}\right)\longrightarrow 0 \hspace{0.1cm}\textnormal{almost surely}.
$$
On the other hand, we have:
$$
\sigma_v^2(1-c_1)\widetilde{\bf h}_i\transconj{\bf V}_2\boldsymbol{\Pi}{\bf V}_2\transconj{\left(\widetilde{\bf h}_i\right)}-\frac{c_1\sigma_v^4(1-c_1)(M-K)}{N\sigma_P^2 }\left[\left(\transconj{\bf H}{\bf H}\right)^{-2}\right]_{i,i}\to 0  \hspace{0.1cm}\textnormal{almost surely}.
$$
Since $\left[\left(\transconj{\bf H}{\bf H}\right)^{-2}\right]-\frac{c_2}{c_2-1}\left[\left(\transconj{\bf H}{\bf H}\right)^{-1}\right]_{i,i}^2\to 0$ by lemma \ref{lemma:new}, we get that:
$$
\sigma_v^2(1-c_1)\widetilde{\bf h}_i\transconj{\bf V}_2\boldsymbol{\Pi}{\bf V}_2\transconj{\left(\widetilde{\bf h}_i\right)}-\frac{\sigma_v^4(1-c_1)c_1c_2}{(c_2-1)}\left[\left(\transconj{\bf H}{\bf H}\right)^{-1}\right]_{i,i}\to 0.
$$
Therefore,
$$
\sigma_{w_d,N}^2-\tilde{\sigma}_{w_d,N}^2\longrightarrow 0 \hspace{0.2cm}  \textnormal{almost surely}
$$
where,
$$
\tilde{\sigma}_{w_d,N}^2=\left(\sigma_v^2(1-c_1)+\frac{c_1(c_2+1)(1-c_1)\sigma_v^4}{(c_2-1)\sigma_{P_d}^2}\right)\left[\left(\transconj{\bf H}{\bf H}\right)^{-1}\right]_{i,i}.
$$
Consequently,
\begin{align*}
\mathbb{E}\left[\exp\left(\jmath\Re\left(z^*\left(\Delta{\bf W}\right)_{i,j}\right)\right)|{\bf H},{\bf W},{\bf V}_2\right]&=\mathbb{E}\left[\exp\left(-\jmath\Re\left(z^*\tilde{\bf w}_i{\bf J}_j+z^*{\bf h}_i^{\#}{\bf V}_2{\bf W}\left({\bf e}_j-{\bf J}_j\right)\right)\right)|{\bf W},{\bf v}_2\right]\\
\times\exp\left(-\frac{1}{4}|z|^2\tilde{\sigma}_{w_d,N}^2\right).
\end{align*}
Conditioning on ${\bf W}$ and ${\bf H}$, $\tilde{\bf w}_i{\bf J}_j+{\bf h}_i^{\#}{\bf V}_2{\bf W}\left({\bf e}_j-{\bf J}_j\right)$ is a Gaussian random variable with mean equal to $\tilde{\bf w}_i{\bf J}_j$ and variance $\sigma_{w_m,N}^2$ given by:
\begin{eqnarray*}
\sigma_{m_d,N}^2&=&\mathbb{E}\left[{\bf h}_i^{\#}{\bf V}_2{\bf W}\left({\bf e}_j-{\bf J}_j\right)\left(\transconj{\bf e}_j-\transconj{\bf J}_j\right)\transconj{\bf W}\transconj{\bf V}_2\transconj{\left({\bf h}_i^{\#}\right)}|{\bf W},{\bf H}\right]\\
&=&\frac{\sigma_v^2}{N\sigma_{P_d}^2}\left[\left(\transconj{\bf H}{\bf H}\right)^{-1}\right]_{i,i}\left(\transconj{\bf e}_j-\transconj{\bf J}_j\right){\bf W}\transconj{\bf W}\left({\bf e}_j-{\bf J}_j\right).
\end{eqnarray*}
Using corollary \ref{corollary:quadratic}, we can easily prove that:
$$
\sigma_{m_d,N}^2-\tilde{\sigma}_{m_d,N}^2\longrightarrow 0 \hspace{0.2cm}\textnormal{almost surely},
$$ 
where
$$
\tilde{\sigma}_{m_d,N}^2=\frac{(1-c_1)\sigma_{w_d}^2\sigma_v^2}{\sigma_{P_d}^2}\left[\left(\transconj{\bf H}{\bf H}\right)^{-1}\right]_{i,i}.
$$
Conditioning only on ${\bf H}$, the conditional characteristic function satisfies:
$$
\mathbb{E}\left[\exp\left(\jmath\Re\left(z^*\left(\Delta{\bf W}_d\right)_{i,j}\right)\right)|{\bf H}\right]-\mathbb{E}\left[\exp\left(-j\Re\left(z^*\tilde{\bf w}_i{\bf J}_j\right)\right)\right]\exp\left(-\frac{1}{4}|z|^2\left(\tilde{\sigma}_{w_d,N}^2+\tilde{\sigma}_{m_d,N}^2\right)\right)\longrightarrow 0 .
$$
Giving the structure of the matrix ${\bf J}$, $\tilde{\bf w}_i{\bf J}_j$ involves the  average of $\frac{1}{c_1}$ symmetric independent and identically distributed discrete random variables,  and therefore,
$$
\mathbb{E}\left[\exp\left(-j\Re\left(z^*\tilde{\bf w}_i\right)\right)\right]=\sum_{i=1}^{\mathcal{Q}}p_i\exp\left(\jmath\Re\left(z^*\alpha_i\right)\right)$$
where $\mathcal{Q}$ is the set of all possible values of $\overline{\bf W}_{i,k}=c_1\sum_{i=1}^{\frac{1}{c_1}}{\bf W}_{i,k}$ and $p_i$ is the probability that $\overline{\bf W}_{i,k}$ takes the value $\alpha_i$.
Consequently;
$$
\mathbb{E}\left[\exp\left(\jmath\Re\left(z^*\left(\Delta{\bf W}_d\right)_{i,j}\right)\right)|{\bf H}\right]=\sum_{i=1}^{\mathcal{Q}}p_i\exp\left(\jmath\Re\left(z^*\alpha_i\right)\right)\exp\left(-\frac{1}{4}|z|^2\left(\tilde\sigma_{m_d,N}^2+\sigma_{w_d,N}^2\right)\right).
$$
We conclude the proof by noting that 
$$
\tilde\sigma_{m_d,N}^2+\sigma_{w_d,N}^2=\sigma_{w_d}^2\left[\left(\transconj{\bf H}{\bf H}\right)^{-1}\right]_{i,i}\delta_d.
$$

\bibliographystyle{bmc_article}
\bibliography{biblio}

\end{document}